\documentclass[12pt]{article}

\usepackage[margin=1in]{geometry}

\usepackage{setspace}
\usepackage{setspace,diagbox}
\usepackage{amssymb, lipsum}
\usepackage{amsmath}
\usepackage{amsthm}
\usepackage{array}
\usepackage{enumitem}
\usepackage{nicefrac}
\usepackage{esint}
\usepackage{float}
\usepackage{physics}
\usepackage{bm}
\usepackage{bbm}
\usepackage[colorlinks]{hyperref}
\usepackage{tcolorbox}
\usepackage{graphicx}
\usepackage{subcaption}
\usepackage{tabularx}
\usepackage{soul}
\usepackage[export]{adjustbox}

\usepackage[parfill]{parskip}
\usepackage{natbib}
\usepackage{authblk}

\usepackage{tikz}
\usetikzlibrary{shapes.geometric, arrows.meta, positioning}

\usepackage{algorithm}
\usepackage[noend]{algpseudocode}
\usepackage{booktabs}

\definecolor{webgreen}{rgb}{0,0.4,0}
\definecolor{webbrown}{rgb}{0.6,0,0}
\definecolor{purple}{rgb}{0.5,0,0.25}
\definecolor{darkblue}{rgb}{0,0,0.7}
\definecolor{darkred}{rgb}{0.7,0,0}
\definecolor{darkgreen}{rgb}{0,0.7,0}

\hypersetup{colorlinks,citecolor=darkred,filecolor=black,linkcolor=darkblue,urlcolor=webgreen}
\newtheoremstyle{definition_style}%
  {}
  {}
  {\normalfont}
  {}
  {\bfseries}
  {.}
  {.5em}
   {\thmname{\bfseries#1}\thmnumber{ \bfseries#2}\thmnote{ (\bfseries#3)}}
  \theoremstyle{definition_style}

\newtheorem{theorem}{Theorem}
\newtheorem{axiom}{Axiom}

\newtheorem{definition}{Definition}
\newtheorem{lemma}{Lemma}
\newtheorem{proposition}{Proposition}

\newtheorem{observation}{Observation}

\newtheorem{example}{Example}

\newtheoremstyle{normaltext}
  {2pt}   
  {2pt}   
  {}      
  {}      
  {\bfseries} 
  {.}     
  { }     
  {}      

\theoremstyle{normaltext}

\title{Stable Matching with Choice Correspondences: Beyond Path Independence}

\author{Varun Bansal\footnote{Economics and Planning Unit, Indian Statistical Institute, Delhi Center, 7, S. J. S. Sansanwal Marg, New Delhi 110016. Email: varun23r@isid.ac.in }\, Mihir Bhattacharya\footnote{Department of Economics, Ashoka University, Rajiv Gandhi Education City, Rai, Sonipat,  Haryana, 131029, India. Email: mihir.bhattacharya@ashoka.edu.in} \,
Ojasvi Khare\footnote{Shiv Nadar University, Department of Economics, Gautam Buddha Nagar, Uttar Pradesh, 201314. Email: ojasvi.khare@snu.edu.in} }

\date{}

\begin{document}

\maketitle                                     

\begin{abstract}

We study stable matching with agents' choice correspondences as primitives. For many-to-many markets we introduce Individually Rational Persistence (IRP), which admits status-quo persistence, and show that substitutability (SUB) and IRP guarantee a stable matching; the constructive proof yields a symmetric matching algorithm with no proposing side and long-term contracts. For one-to-one markets we introduce replacement-based stability and show that SUB together with binary acyclicity (BA) guarantees existence. Combining the two, in the many-to-one setting, where SUB and irrelevance of rejected contracts (IRC) guarantee stability, we show that SUB with IRP on firms and BA on workers guarantees existence. We show that, in the presence of SUB, IRC is not necessary for existence in any of the three models, so stable matchings exist outside the domain of path independence.

\end{abstract}
 
\section{Introduction}

The matching literature typically takes agents' preferences to be complete, transitive orderings over outcomes, from which stable matchings are derived. This framework is useful, but its informational and behavioral demands are strong: preferences may be unknown, and observed choices may display bounded rationality or other non-standard choice behavior.\footnote{Such choices may fail to be representable by a complete and transitive binary relation. A binary relation $R$ over $X$ is (i) complete if, for all $x,y\in X$, either $xRy$ or $yRx$; and (ii) transitive if, for all $x,y,z\in X$, $[xRy \text{ and } yRz] \implies xRz$.} Instead, we model agents directly through choice correspondences, taking as primitive the observable choices agents make from feasible sets of contracts. This admits behavior that no strict ranking captures, such as status-quo persistence: an agent retains a previously accepted contract even if a superior contract arrives later, as in  employment-protection institutions, seniority, contractual commitments, or relationships.\footnote{Status-quo persistence is a well-documented departure from preference-maximizing choice,
first identified experimentally by \cite{samuelson1988status} and since studied extensively (\cite{MasatliogluOk2005} and \cite{ApesteguiaBallester2013}).}

With agents represented by choice correspondences, stability is formulated in terms of choices rather than preferences. Following \cite{chambers2017choice}, a matching is stable if each agent chooses its assigned contracts from the matching (individual rationality) and no set of contracts can be added such that every agent involved would choose those contracts alongside its current holdings (no blocking).\footnote{A set is individually rational if the agent's choice correspondence returns the entire set when it is offered.} They show that path independence (PI), equivalent to substitutability (SUB) together with irrelevance of rejected contracts (IRC), is sufficient for the existence of a CY-stable many-to-many matching.\footnote{PI requires choices to be consistent across decompositions of a menu: choosing from each part and then from the union of the resulting sets yields the same outcome as choosing from the whole menu.} We retain SUB but replace IRC with Individually Rational Persistence (IRP), and show that SUB and IRP together suffice for existence. Thus, our approach provides an alternative to that of \cite{chambers2017choice}.\footnote{We show in the Appendix that IRC and IRP are logically independent (Proposition \ref{prop:irc-irp}).}

IRP requires that adding a single contract to an individually rational set cannot displace any of its existing members if the new contract is chosen alongside the set. However if the new contract is not chosen, IRP imposes no restriction. Thus, IRP protects established relationships without requiring the cross-menu consistency imposed by IRC, making it natural for markets organized around long-term contracts. IRP thus formalizes status-quo persistence without the cross-menu
consistency IRC imposes (Example~\ref{ex:statusquo}). This is more than an alternative axiomatization: since IRP is logically
independent of IRC (Proposition~\ref{prop:irc-irp}), SUB and IRP guarantee
stability with choice behavior that IRC and the
preference-based approach rule out.\footnote{Every path-independent
correspondence is rationalizable by a ranking over sets
\cite{yang2020rationalizable}; those satisfying SUB and IRP need not be
(Appendix~\ref{app:comparison}). This connect our 
domain to the literature on boundedly rational and menu-dependent choice
\citep{RubinsteinSalant2006, Gerasimou2018}.}

In many matching environments, an established relationship carries value that a
nominally superior outside option may not. Hiring, doctor--hospital assignment,
research collaborations, and long-term supply agreements all involve
relationship-specific investment, trust, and adaptation, so that once a
relationship is formed, its realized value can exceed the prospective value of a
newly available alternative; switching is then costly, and incumbents are
protected.\footnote{This is the reduced-form implication for observed choice of
relationship-specific investment and switching costs \cite{klein1978vertical,
williamson1985economic, klemperer1995competition}, and of seniority and
implicit-contract arrangements in labor markets \cite{baily1974wages,
azariadis1975implicit}.} Academic tenure is the limiting case: once granted, an
appointment is not revoked when a stronger candidate later appears. IRP captures
this incumbency as a static, cross-menu restriction on choice; applied in our
algorithm's processing order, it yields the temporal property in which an accepted
contract survives a later arrival. Although this behavior lies outside the
preference-based domain, it remains compatible with stable matching.

We study all three matching formats, each with its own notion of stability:
CY-stability for many-to-many \cite{chambers2017choice}, R-stability for
one-to-one, and MO-stability for the many-to-one setting of
\cite{hatfield2005matching}. In each case we retain SUB and replace IRC---by IRP
where an agent holds many contracts, by BA where it holds one; the many-to-one
market has both kinds of side, so its result combines IRP on firms with BA on
workers.

We provide the Symmetric Matching Algorithm (SMA) for many-to-many matching
under choice correspondences. It processes contracts one at a time, accepting a
contract if and only if both of its endpoints choose it alongside their current
holdings and rejecting it permanently otherwise. Unlike the Deferred Acceptance
Algorithm, SMA has no proposing and responding sides: firms and workers are
treated symmetrically, and there is no recall or re-proposing, which makes it
operationally simpler. Under SUB and IRP, it terminates in a CY-stable matching.

For one-to-one matching, we introduce a replacement-based notion of stability
together with a binary acyclicity (BA) condition on pairwise choices, and show
that SUB and BA suffice for the existence of a stable matching. BA constrains
only strict pairwise choice and, under SUB, is strictly weaker than IRC: every
path-independent correspondence satisfies BA but not conversely, so existence
holds on a strictly larger domain than PI.

Our third result unifies the two. In many-to-one markets, firms hold many
contracts while each worker holds at most one, so the two sides play
structurally different roles: an added contract enlarges a firm's holding,
whereas a worker can only replace her current match. We introduce MO-stability,
combining set-inclusion blocking on firms with replacement blocking on workers,
and show that SUB with IRP on firms and BA on workers guarantees existence. This
is the setting of \cite{hatfield2005matching}, where \cite{aygun2013matching}
show that SUB alone does not suffice and that IRC is implicitly assumed. IRC is
not necessary even here: IRP is independent of IRC, and BA is weaker than IRC
under SUB, so stability again lies outside PI.

\subsection{Literature Review}
\cite{hatfield2005matching} introduce the matching-with-contracts framework and show that substitutability (SUB) of agents' choice functions guarantees the existence of a stable allocation in many-to-one markets, building on the substitutes condition of \cite{kelsocrawford1982} and the revealed-preference approach of \cite{roth1984} and \cite{blair1988lattice}. \cite{aygun2013matching} subsequently show that this existence result implicitly relies on a second condition, irrelevance of rejected contracts (IRC), under which removing a rejected contract from a menu never changes the chosen set; without IRC, SUB alone does not guarantee existence even there. Together, SUB and IRC are equivalent to PI (\cite{aizerman1981general}), which has since become the standard sufficient condition for stable matching with contracts (\cite{hatfield2010substitutes, hatfield2017contract}).

A large literature weakens the conditions behind PI. On the substitutability side, \cite{hatfield2010substitutes} restrict the scope of complementarities through unilateral substitutability, and \cite{hatfield2017contract} introduce completability, requiring that a choice correspondence admit an extension to a fully substitutable choice structure; the latter studies many-to-many matching with contracts under SUB and IRC, focusing on contract design and the structure of the stable set rather than on alternatives to IRC. A related strand weakens substitutability itself: \cite{zhang2016} introduces observable substitutability, strengthened by \cite{bandohiraizhang2021substitutes} to weak observable substitutability across doctors (WOSAD), while \cite{kadam2017} shows that unilateral substitutability suffices for a substitutable completion to exist. \cite{hatfieldkominerswestkamp2021} obtain stability and strategy-proofness in many-to-one markets under observable substitutability and size monotonicity.

\textbf{Relation to dynamic matching.} Our algorithm is sequential rather than dynamic in the usual sense: agents' holdings evolve over a sequence of rounds, and IRP governs how already-accepted contracts are protected against later arrivals. This connects the paper to a strand of the matching literature concerned with sequential decision-making under commitment. \cite{kurino2014} and \cite{blochcantala2017} study dynamic assignment problems in which agents arrive over time and mechanisms must decide, period by period, whether to match immediately or delay; \cite{doval2022} develops a notion of dynamic stability for markets in which agents enter one at a time and matches once formed cannot be costlessly undone. Our model also differs from the literature on dynamic matching under stochastic arrivals and departures (e.g., \cite{baccaraleeyariv2020}; \cite{akbarpourligharan2020}), where the central question is a timing trade-off between matching now and waiting for a better future partner under discounting or scarcity. IRP's permanence property is closer in spirit to the commitment-without-recall concerns of the former branch, even though our primitives are static.

IRP is motivated by status-quo persistence, one of several documented departures from preference-maximizing choice.\footnote{See \cite{HuberPaynePuto1982}, \cite{MasatliogluOk2005}, \cite{RubinsteinSalant2006}, \cite{ApesteguiaBallester2013}, \cite{SalantSiegel2018}, and \cite{Gerasimou2018} on the attraction effect, status-quo bias, and choice overload.} Because this pattern lies outside path independence, a recent literature asks how far the PI assumption underlying \cite{chambers2017choice} can be relaxed while retaining stable matchings. In the one-to-one setting, \cite{CaspariKhanna2025} characterize stable mechanisms using binary menu conditions weaker than PI; our binary acyclicity condition belongs to this family, though we obtain replacement stability by direct construction rather than through incentive-compatible design. \cite{Kuvalekar2022} treats one-to-one matching with incomplete preferences under a weak stability notion.

Closest to our many-to-many result is \cite{BandoHiraiImamura2025}, who also weaken IRC while retaining SUB for existence in this setting. Their condition replaces IRC by a monotonicity requirement, whereas ours, IRP, is a persistence requirement on individually rational sets; the two are logically independent (Appendix~\ref{app:comparison}), so neither result subsumes the other. Our approach differs further in method and scope: our existence proof is constructive, via a symmetric matching algorithm with no proposing side, and we extend the analysis to one-to-one and many-to-one markets, obtaining existence in each setting.

\textbf{Consistency-type axioms on choice correspondences.} A related literature studies conditions restricting how choices vary across nested menus. \cite{alkan2002class} develops a many-to-many matching model under PI, and \cite{EcheniqueYenmez2015} introduce the closely related irrelevance of rejected students (IRS), a school-choice analogue of IRC. These conditions require invariance when rejected alternatives are removed from a menu; IRP instead requires persistence when a menu expands. This drives the monotone holdings, permanent acceptance, and the symmetry of our algorithm.

The paper is organized as follows. Section \ref{sec2} introduces the general matching model, while Section \ref{sec3} presents the axioms on choice correspondences used in the many-to-many model, and provides the main result on stability along with the matching algorithm.  Section \ref{sec4} studies the one-to-one matching problem and establishes existence results with a constructive algorithm, using a binary acyclicity condition.
Section \ref{sec5} combines the two conditions to treat the many-to-one problem. All proofs omitted from the main text appear in the Appendix. 

\section{Model}\label{sec2}
There is a finite set of workers $W$, firms $F$, and contracts $X$. Each contract $x \in X$ identifies a unique firm--worker pair together with a comprehensive set of employment terms, which may include wages and other contractual benefits. For each contract $x \in X$, the firm and worker associated with the contract are denoted by $f(x)$ and $w(x)$, respectively.\footnote{We allow a bilateral contract to be formed between a worker and a firm, but we do not require it to be unique: it is possible to have $x \neq x'$ with $f(x)=f(x')$ and $w(x)=w(x')$.} 

For each worker $w$, let $X_w = \{x \in X : w(x) = w\}$ be the contracts associated with $w$, and for each firm $f$, let $X_f = \{x \in X : f(x) = f\}$ be the contracts associated with $f$. For an agent $i \in W \cup F$, we write $X_i$ for $X_w$ if $i=w$ and for $X_f$ if $i=f$.

Each agent $i$ has a choice correspondence
$C_i : 2^{X_i} \to 2^{X_i}$ with $C_i(S) \subseteq S$ for all $S \subseteq X_i$;
choices may be empty. We assume an agent's choice depends only on the contracts
available to them: for any $S \subseteq X$, $C_i(S) = C_i(S \cap X_i)$. A matching market is a tuple $\langle W, F, X, (C_w)_{w \in W}, (C_f)_{f \in F} \rangle$.

\begin{definition}[Matching]\label{def:matching}
A \emph{matching} $\nu \subseteq X$ is a set of contracts. For each worker $w \in W$, let $\nu_w = \{x \in \nu : w(x) = w\}$ denote the contracts assigned to $w$ under $\nu$; for each firm $f \in F$, let $\nu_f = \{x \in \nu : f(x) = f\}$ denote the contracts assigned to $f$.
\end{definition}

A matching is many-to-many by default: any set of firms and workers can be jointly matched through any set of contracts, with no capacity constraint beyond the market itself.

\section{The many-to-many matching problem}\label{sec3}
We use the notion of stability introduced by \cite{chambers2017choice}, which we refer to throughout as CY-stability.

\begin{definition}[CY-stability]
A matching $\nu$ is \emph{individually rational (IR)} if $C_i(\nu_i) = \nu_i$ for every agent $i$. A non-empty set $Z \subseteq X \setminus \nu$ \emph{blocks} $\nu$ if, for every $i$ with $Z_i := Z \cap X_i \ne \emptyset$, $Z_i \subseteq C_i(\nu_i \cup Z_i)$. A matching is \emph{CY-stable} if it is IR and admits no blocking set.
\end{definition}

A matching $\nu$ is individually rational (IR) if every agent chooses all contracts assigned to her from her current holdings. It is blocked by a non-empty set of new contracts $Z$, disjoint from $\nu$, if every agent involved in $Z$, when offered her current holdings $\nu_i$ together with the new contracts $Z_i$, chooses all contracts in $Z_i$. CY-stability rules out such blocking sets: no collection of unmatched contracts can be jointly chosen by all agents it involves in addition to their current holdings.

We provide the axioms for choice correspondences. The first one is a consistency requirement which allows for substitutability between workers (or firms).

\begin{axiom}[Substitutability (SUB)] A choice correspondence $C_{i}$ satisfies \text{SUB} if for every $A,B\subseteq X_{i}$ such that  $A \subseteq B$, then $C_{i}(B)\cap A\subseteq C_{i}(A)$. \end{axiom}

Substitutability (SUB) requires contracts to behave as gross substitutes: if an agent chooses a contract from a larger set of available contracts, the contract remains chosen from any smaller set that contains it. This is the substitutes condition underlying the classical theory of matching with contracts (e.g., \cite{kelsocrawford1982}, \cite{hatfield2005matching}). It rules out complementarities between contracts, for example a firm cannot choose worker $A$ only when it can also retain worker $B$, and then drop $A$ when $B$ becomes unavailable.\footnote{\cite{kelsocrawford1982} introduced the gross substitutes condition in a job matching model with wages and preferences over sets of contracts to establish the existence of stable outcomes. \cite{roth1984} showed that stable many-to-one matchings exist under responsive preferences, which imply SUB. Subsequent work has adopted SUB directly as a condition on choice functions to guarantee stability, including \cite{hatfield2005matching}, \cite{hatfield2010substitutes}, \cite{chambers2017choice}, and \cite{pycia2023matching}. An alternative formulation of the axiom is as follows: a choice correspondence $C:2^{X}\to 2^{X}$ satisfies SUB if, for every $A,B\subseteq X$ such that for every $a\in A\subseteq B$, $a\in C(B)$ implies $a\in C(A)$.} The following observations are immediate.

\begin{observation}\label{obs:subcc}
Suppose $C_i$ satisfies SUB. Then:
\begin{enumerate}[label=(\roman*)]
\item If $x\in A\subseteq B$ and $x\notin C_i(A)$, then $x\notin C_i(B)$.
\item $C_i(S)$ is IR for every $S\subseteq X_i$; that is, $C_i(C_i(S))=C_i(S)$.
\end{enumerate}
\end{observation}

We now define the notion of a challenger and introduce our new axiom.

\begin{definition}[Challenger]
For an IR set $A$ and $x\in X_i\setminus A$, $x$ is a \emph{challenger} to $A$ if $x\in C_i(A\cup\{x\})$.
\end{definition}

\begin{axiom}[Individually Rational Persistence, IRP]\label{ax:IRP}
A choice correspondence $C_i$ satisfies IRP if, for every IR set $S\subseteq X_i$ and every challenger $x\notin S$ to $S$,
\[
S\subseteq C_i(S\cup\{x\}).
\]
\end{axiom}

IRP requires that an IR set cannot be partially discarded when a single new contract is added: if the new contract is a challenger, it must be accepted without displacing any existing contract; otherwise IRP imposes no restriction on the choice from the expanded set. The intuition is analogous to tenure: once a professor has been appointed under the department's standard, the appointment is not revoked if a stronger candidate arrives.

The following example exhibits a choice correspondence satisfying IRP for which both SUB and IRC fail.

\begin{example}[IRP without SUB or IRC]
Let $X_i = \{a,b\}$ with $C_i(\{a\}) = \emptyset$, $C_i(\{b\}) = \emptyset$, and $C_i(\{a,b\}) = \{a\}$. Both SUB and IRC fail: taking $A = \{a\} \subseteq B = \{a,b\}$, $C_i(B) \cap A = \{a\} \not\subseteq C_i(A) = \emptyset$, violating SUB; and $C_i(B) = \{a\} \subseteq A$ yet $C_i(A) \ne C_i(B)$, violating IRC. IRP holds vacuously: the only IR set is $\emptyset$, and neither $a$ nor $b$ is a challenger to it, since $a \notin C_i(\{a\})$ and $b \notin C_i(\{b\})$.
\end{example}

The next example provides a choice correspondence that violates IRP.

\begin{example}\label{eg:scope}
Let $X_i = \{x,y\}$ with $C_i(\{x\}) = \{x\}$ and $C_i(\{x,y\}) = \{y\}$. Then $\{x\}$ is IR, and $y$ is a challenger to it since $y \in C_i(\{x,y\})$. IRP would require $\{x\} \subseteq C_i(\{x,y\}) = \{y\}$, which fails since $x \notin \{y\}$. The arrival of $y$ displaces the previously accepted contract $x$ entirely, violating IRP.
\end{example}

The following example provides a correspondence satisfying SUB and IRP that produces the status-quo effect, with the incumbent contract protected by arrival sequence in the matching algorithm (Algorithm \ref{alg:main}).

\begin{example}[IRP and the status-quo effect]\label{ex:statusquo}
Let $X_i=\{b,c\}$ with
\[
C_i(\{b\})=\{b\},\qquad C_i(\{c\})=\{c\},\qquad C_i(\{b,c\})=\varnothing.
\]
This correspondence satisfies SUB and IRP but violates IRC. The two contracts are individually acceptable but cannot be held together. When contracts are processed one at a time (Algorithm~\ref{alg:main}), the first to arrive is accepted and thereafter retained: if $b$ is processed first the agent holds $\{b\}$ and later rejects $c$, and symmetrically if $c$ arrives first. The incumbent is never displaced by the later arrival: a status-quo effect. Note that the correspondence itself is symmetric in $b$ and $c$; the asymmetry is produced by arrival sequence. A rule that instead named a fixed winner, $C_i(\{b,c\})=\{b\}$, would violate IRP, since from the IR set $\{c\}$ the challenger $b$ would displace the incumbent. Because IRC fails, the correspondence is not rationalizable by any ranking over sets (Appendix~\ref{app:comparison})---as it must be, since a fixed preference would select the same contract from $\{b,c\}$ regardless of history. 
\end{example}

\subsection{The Symmetric Matching Algorithm (SMA)}\label{sec:algorithm}

We now describe the Symmetric Matching Algorithm (SMA). It processes the contracts in $X$ one at a time, in a fixed but arbitrary order, and maintains for each agent a set of contracts that agent has accepted so far, which we call the agent's \emph{holding}.  

\textbf{Holdings}: Fix an enumeration $x_1, \dots, x_n$ of $X$, where $n = |X|$; the order is arbitrary. For each agent $i$ and each round $t \in \{0, 1, \dots, n\}$, let $H_i(t) \subseteq X_i$ denote agent $i$'s \emph{holding} after round $t$: the contracts naming $i$ that have been accepted through round $t$. Initialize $H_i(0) := \emptyset$ for every agent $i$.

\textbf{Processing a contract}: At round $t$, the single contract $x_t$ is processed. Write $w := w(x_t)$ and $f := f(x_t)$ for its two endpoints\footnote{By endpoints, we mean the worker and firm associated with the contract.}. The contract is \emph{accepted} if each endpoint would individually choose it when it is added to that endpoint's current holding, i.e., if
\begin{equation}\label{eq:accept}
x_t \in C_w\!\big(H_w(t-1) \cup \{x_t\}\big)
\quad\text{and}\quad
x_t \in C_f\!\big(H_f(t-1) \cup \{x_t\}\big).
\end{equation}
If \eqref{eq:accept} holds, both endpoints update their holdings simultaneously:
\[
H_w(t) := C_w\!\big(H_w(t-1) \cup \{x_t\}\big),
\qquad
H_f(t) := C_f\!\big(H_f(t-1) \cup \{x_t\}\big).
\]
If \eqref{eq:accept} fails, $x_t$ is rejected and both endpoints retain their holdings: $H_w(t) := H_w(t-1)$ and $H_f(t) := H_f(t-1)$. Every agent $j \notin \{w, f\}$ is not an endpoint of $x_t$ and is unaffected: $H_j(t) := H_j(t-1)$. 

The \textbf{terminal matching} is $\nu_i^\ast := H_i(n)$ for every agent $i$; by construction $x \in \nu^\ast$ iff $x \in H_w(n)$ and $x \in H_f(n)$ simultaneously, since the two update only together.

Therefore, each contract is tested individually against the current holdings of its two agents, and the cumulative effect of these choices determines the final matching. We present the formal description of the Symmetric Matching Algorithm (SMA) in Algorithm \ref{alg:main}.

\begin{algorithm}[H]
\caption{Symmetric Matching Algorithm (SMA)} \label{alg:main}
\begin{algorithmic}[1]
\State Fix an enumeration $x_1, \dots, x_n$ of $X$, where $n = |X|$
\State $H_i(0) \gets \emptyset$ for every agent $i$
\For{$t = 1$ \textbf{to} $n$}
  \State $w \gets w(x_t)$, \; $f \gets f(x_t)$
  \If{$x_t \in C_w(H_w(t-1) \cup \{x_t\})$ \textbf{and} $x_t \in C_f(H_f(t-1) \cup \{x_t\})$}
    \State $H_w(t) \gets C_w(H_w(t-1) \cup \{x_t\})$ \Comment{accept: update both endpoints together}
    \State $H_f(t) \gets C_f(H_f(t-1) \cup \{x_t\})$
  \Else
    \State $H_w(t) \gets H_w(t-1)$, \; $H_f(t) \gets H_f(t-1)$ \Comment{reject permanently}
  \EndIf
  \State $H_j(t) \gets H_j(t-1)$ for every agent $j \notin \{w, f\}$
\EndFor
\State \Return the terminal holdings $\big(H_i(n)\big)_{i \in F \cup W}$
\end{algorithmic}
\end{algorithm}

Our algorithm differs from classical DAA in two respects. First, it has no proposing side: every contract is evaluated against both endpoints' current holdings, eliminating any structural proposer advantage. Second, acceptance is permanent rather than tentative. A contract enters $\nu^\ast$ only after both endpoints choose it and is never subsequently displaced (Lemma~\ref{lem:invariant}).

The definition of stability allows a blocking deviation to consist of an entire set of contracts. Our analysis, however, proceeds by testing one contract at a time. The following lemma shows that this is without loss of generality.

\begin{lemma}[Reduction of Blocking Sets]\label{lem:reduction}
Suppose $C_i$ satisfies SUB. Let $\nu_i \subseteq X_i$ be IR, and let $Z_i \subseteq X_i$ satisfy $Z_i \cap \nu_i = \emptyset$ and $Z_i \subseteq C_i(\nu_i \cup Z_i)$. Then $x \in C_i(\nu_i \cup \{x\})$ for every $x \in Z_i$.
\end{lemma}

\begin{proof}
Fix $x \in Z_i$. Since $\nu_i \cup \{x\} \subseteq \nu_i \cup Z_i$, SUB gives $C_i(\nu_i \cup Z_i) \cap (\nu_i \cup \{x\}) \subseteq C_i(\nu_i \cup \{x\})$. Since $x \in Z_i \subseteq C_i(\nu_i \cup Z_i)$ and $x \in \nu_i \cup \{x\}$, $x$ lies in the left-hand intersection, hence $x \in C_i(\nu_i \cup \{x\})$.
\end{proof}

We now present our main result. 

\subsection{The main result (many-to-many)}\label{sec:mainresults}

Since the set of contracts $X$ is finite, the algorithm examines each contract exactly once and therefore terminates after $|X|$ rounds. The following theorem shows that its output is CY-stable.

\begin{theorem}\label{thm:stability}
Suppose every agent's choice correspondence, firm or worker, satisfies SUB and IRP. Then the terminal matching obtained by the SMA (Algorithm \ref{alg:main}) is CY-stable.
\end{theorem}

The proof, together with a discussion of the roles of SUB and IRP, is provided in the Appendix. Theorem~\ref{thm:stability} is the paper's main existence result: substitutability together with IRP guarantees the existence of a CY-stable matching. Thus, PI, and in particular IRC, can be replaced by IRP for existence of CY-stable matchings. 

The proof rests on a key structural property of the algorithm: each agent's holding grows monotonically throughout its execution. SUB and IRP play complementary roles in establishing this property and, ultimately, stability. Lemma~\ref{lem:invariant} shows that holdings never shrink. Because each holding is a chosen set, SUB (Observation~\ref{obs:subcc}(ii)) ensures that it is individually rational at every round, allowing IRP to be applied whenever a new contract is considered. Since contracts are introduced one at a time, IRP prevents a newly accepted contract from displacing an incumbent.

Monotonicity also makes rejection permanent. If a contract fails an agent's choice test at some round, all subsequent menus are supersets of the menu from which it was rejected. SUB therefore ensures that the contract remains rejected thereafter (Observation~\ref{obs:permrej}). Since a contract is accepted only when both endpoints choose it, rejection by either endpoint permanently excludes it.

Lemma~\ref{lem:srirc} converts this permanence into the condition needed for stability: every contract absent from the terminal matching $\nu^\ast$ is rejected by at least one endpoint when evaluated against its terminal holding. Individual rationality of $\nu^\ast$ follows because every terminal holding is a chosen set. Moreover, by Lemma~\ref{lem:reduction}, if a set $Z$ blocked $\nu^\ast$, every contract in $Z$ would have to be IR for both endpoints against their terminal holdings, contradicting Lemma~\ref{lem:srirc}. Hence, $\nu^\ast$ is CY-stable.

The framework also permits choice correspondences that return the empty set. If $C_i(S)=\emptyset$ for some feasible menu $S$, then no contract in $S$ is chosen by agent $i$. In particular, if $C_i(\nu_i\cup{x})=\emptyset$, then $x$ is rejected on $i$'s side and, by Observation~\ref{obs:permrej}, remains rejected as menus expand. Thus, allowing empty choices does not affect the analysis.

We turn next to the one-to-one case. CY-stability is not natural when each agent can hold at most one contract: with $|\nu_i|\leq 1$, any profitable deviation must replace the current contract rather than add to it. We therefore adopt a replacement-based stability notion (R-stability, Definition~\ref{def6}) and show that SUB together with BA guarantees the existence of an R-stable matching. Section~\ref{sec5} then combines this with IRP to treat the many-to-one case.

\section{The one-to-one matching problem}\label{sec4}

Let $| F| = | W| = m$. In this section we restrict attention to one-to-one matchings, i.e., $|\nu_{i}|\leq 1$ for all $i\in {F}\cup {W}$. We define it formally.

\begin{definition}[One-to-One Matching] 
A \text{one-to-one matching} is a matching $\nu\subseteq {X}$ such that for every agent $i \in {F} \cup {W}
$, $\left|\nu_i\right| \leq 1$. Therefore, in a one-to-one matching each agent is involved in at most one contract. 

\end{definition}

Note that this also implies that each firm is uniquely matched to a worker (and vice versa) through one contract. However, the set of contracts could involve multiple contracts between the same worker-firm pair. We provide an axiom which is used to prove the existence of stable one-to-one matchings.

\begin{axiom}[Binary acyclicity (BA)]\label{ax:ba}
Let $x\,P\,y$ denote $C(\{x,y\})=\{x\}$. The correspondence $C$ satisfies binary
acyclicity if $P$ has no cycle: there is no sequence $x_1,\dots,x_k$ with
$x_1\,P\,x_2,\ \dots,\ x_{k-1}\,P\,x_k,\ x_k\,P\,x_1$.
\end{axiom}

\noindent BA is strictly weaker than transitivity of pairwise choice: it forbids strict
cycles but permits incomparabilities, e.g.\ $x\,P\,y$ and $y\,P\,z$ with
$C(\{x,z\})=\{x,z\}$.

CY-stability is based on set-inclusion blocking: a set of contracts blocks a matching when every agent it involves would choose those contracts \emph{alongside} its current holdings, so that blocking expands an agent's assignment. However, when $|\nu_i|\le 1$, a matched agent cannot take on an additional contract; its only possible deviation is to give up its current match in exchange for another. This replacement-based notion of blocking is similar to that in \cite{hatfield2005matching} and is more appropriate here, since every profitable deviation necessarily takes the form of complete replacement. We provide the formal definition below.

\begin{definition}[R-stability]\label{def6}
A matching $\nu \subseteq X$ is \emph{R-stable} if it is \emph{individually rational (IR)}: $C_i(\nu_i) = \nu_i$ for every agent $i \in F \cup W$, and admits no \emph{blocking contract}: there is no $x \notin \nu$ with $C_{f(x)}(\nu_{f(x)} \cup \{x\}) = \{x\}$ and $C_{w(x)}(\nu_{w(x)} \cup \{x\}) = \{x\}$.
\end{definition}

Therefore, stability in the one-to-one matching problem requires two properties: (i) the choice from the matched set should be the set itself and (ii) no other firm--worker pair should deviate by `replacing' their current contract with another contract involving them.

\subsection{Result (one-to-one)}

\begin{theorem}\label{thm2}
Suppose every agent's choice correspondence satisfies SUB and BA. Then the terminal matching obtained by the matching algorithm \ref{alg:2} is R-stable.
\end{theorem}

We provide a sketch of the proof of this theorem. The argument has three steps: BA yields a preference order (extended) for each agent, deferred acceptance turns these orders into a matching, and SUB certifies that the matching is
stable in the sense R-stability requires. For each agent $i$, define a strict pairwise
relation by $x P_i x'$ whenever $C_i(\{x,x'\}) = \{x\}$. Binary acyclicity implies that this relation has no cycles, so it extends to a strict linear
order over $i$'s IR contracts; pairs on which $i$ chooses both or neither are ties, broken arbitrarily without affecting what follows. Running the firm-proposing deferred acceptance algorithm on these orders produces a one-to-one matching $\nu$, since each agent holds at most one contract. Individual rationality is immediate: firms propose, and workers hold, only IR contracts, so each $\nu_i$ is either empty or a self-chosen singleton, whence $C_i(\nu_i)=\nu_i$.

It remains to rule out blocking. R-stability tests blocking through actual choices, $C_i(\nu_i\cup\{x\})=\{x\}$, whereas deferred acceptance only guarantees the absence of blocks stated through the constructed order. To bridge the two, suppose a contract $x=(f,w)\notin\nu$, R-blocks $\nu$. On each side $C_i(\nu_i\cup\{x\})=\{x\}$ means $x$ is chosen from the pair $\nu_i\cup\{x\}$, and SUB (Observation~\ref{obs:subcc}, applied to the singleton $\{x\}$) then forces $x\in C_i(\{x\})$, so $x$ is IR for $i$ and therefore appears in $i$'s order. Given IR, $C_i(\{x,\nu_i\})=\{x\}$ yields $xP_i\nu_i$ on both sides, so $x$ is a blocking pair for $\nu$ under the constructed orders---contradicting the stability of the deferred-acceptance outcome. Hence no blocking contract exists and $\nu$ is R-stable.

The role of SUB is to connect pairwise comparisons to individual acceptability: it ensures that any contract that can R-block a matching is individually acceptable to both of its endpoints (IR). This bridges the choice-based notion of blocking in R-stability and the order-based exclusion of blocking underlying deferred acceptance. BA alone does not provide this link: a contract may prevail in the relevant pairwise comparisons without being individually acceptable, and hence may block without being capable of replacing an existing contract. SUB is therefore essential for translating binary acyclicity into the absence of profitable replacements.

An implication of Theorem~\ref{thm2} and Proposition~\ref{prop:ba-irc} is that PI is not necessary for R-stability. Proposition~\ref{prop:ba-irc} exhibits a correspondence satisfying SUB and BA but not IRC, hence violating PI, while Theorem~\ref{thm2} guarantees an R-stable matching whenever all agents satisfy SUB and BA. Thus SUB and BA yield existence even outside PI.

The one-to-one and many-to-many results isolate two distinct mechanisms. In the many-to-many market (Section~\ref{sec3}), IRP drives existence: holdings expand monotonically, acceptance is persistent, and set-inclusion blocking is the appropriate stability criterion. In the one-to-one market (Section~\ref{sec4}), BA provides the pairwise ordering for deferred acceptance while SUB ensures that acceptances are IR. The many-to-one market combines these mechanisms: firms hold multiple contracts and satisfy IRP, while capacity-one workers are governed by BA. The next section shows that SUB with IRP on firms and BA on workers guarantees existence of stable matchings in this setting.

\section{The many-to-one matching problem}\label{sec5}

We now study the many-to-one case, in which workers are capacity-one but firms are not. A matching $\nu$ is \emph{many-to-one} if $|\nu_w|\le 1$ for every worker $w\in W$, with no constraint on firms. Therefore, on the worker side the situation is that of the one-to-one problem: a matched worker can only exchange her contract for another, not add to it, so replacement blocking ($C_w(\nu_w\cup\{x\})=\{x\}$) is relevant, as in R-stability (Section~\ref{sec4}). On the firm side, by contrast, firms hold many contracts, an added contract expands the holding, and set-inclusion blocking remains appropriate, as in CY-stability (Section~\ref{sec3}). MO-stability combines the two: inclusion blocking on firms, replacement blocking on workers.

\begin{definition}[MO-stability]\label{def:mo}
A many-to-one matching $\nu$ is \emph{MO-stable} if it is individually rational ($C_i(\nu_i)=\nu_i$ for every agent $i$) and admits no blocking contract: there is no $x\notin\nu$, with $f=f(x)$ and $w=w(x)$, such that
\[
x\in C_f\!\left(\nu_f\cup\{x\}\right)
\qquad\text{and}\qquad
C_w\!\left(\nu_w\cup\{x\}\right)=\{x\}.
\]
\end{definition}

By Lemma~\ref{lem:reduction}, on the
firm side single-contract blocking is without loss of generality due to SUB, so MO-stability rules out all set blocks a many-to-one market admits. The matching algorithm is updated accordingly.

\subsection{The algorithm}

Workers propose; firms accept permanently. Each worker's proposal order is the linear order $\succ_w$ obtained from BA, restricted to acceptable contracts; ties are broken arbitrarily as in the proof of Theorem~\ref{thm2}.

\begin{algorithm}[H]
\caption{Worker-proposing matching with permanent firm acceptance}\label{alg:2}
\begin{algorithmic}[1]
\State $H_f \leftarrow \emptyset$ for every firm $f$
\State for every worker $w$, let $L_w$ be the list of $w$'s acceptable contracts in $\succ_w$-order
\While{some worker $w$ holds no accepted contract and $L_w$ is non-empty}
  \State $x \leftarrow$ the $\succ_w$-maximal contract remaining in $L_w$; let $f \leftarrow f(x)$
  \State remove $x$ from $L_w$
  \If{$x \in C_f(H_f \cup \{x\})$}
     \State $H_f \leftarrow C_f(H_f \cup \{x\})$ \Comment{by IRP, $H_f$ grows: prior contracts retained}
     \State $w$ is now matched to $x$ \Comment{permanent: $x$ is never displaced from $H_f$}
  \Else
     \State $x$ is permanently rejected by $f$ \Comment{by SUB, $x \notin C_f(H_f' \cup \{x\})$ for all later $H_f' \supseteq H_f$}
  \EndIf
\EndWhile
\State \Return $\nu = \bigcup_f H_f$
\end{algorithmic}
\end{algorithm}

Because a worker proposes best-first and stops at its first acceptance, at most one
contract per worker ever enters any firm's holding, so $|\nu_w|\le 1$: the output is
a many-to-one matching. The process terminates since each of the $|X|$ contracts is
proposed at most once.

We show that IRP on firms and BA on workers suffices for existence of MO-stable matchings placing our result outside the domain of PI.

\subsection{Result (many-to-one)}

\begin{theorem}\label{thm:mo}
Suppose every firm's choice correspondence satisfies SUB and IRP, and every worker's
choice correspondence satisfies SUB and BA. Then the terminal matching obtained by the Algorithm \ref{alg:2} is MO-stable.
\end{theorem}

Theorem~\ref{thm:mo} is the counterpart, in the setting where the matching-with-contracts literature is placed, of our many-to-many result. \cite{aygun2013matching} show that SUB alone is insufficient there: without IRC, substitutability does not guarantee a stable allocation. By Proposition~\ref{prop:irc-irp}, an IRP correspondence need not satisfy IRC, so Theorem~\ref{thm:mo} provides existence even where IRC, hence PI, fails. The contrast with the cumulative-offer process is instructive.  There a firm chooses from the entire accumulated menu, and ruling out blocks requires that a contract rejected from a large menu stay rejected on the smaller menu, $\nu_f\cup\{x\}$, this consistency is provided by IRC. Our firms instead process proposals one at a time, so IRP, a single-challenger condition, keeps holdings monotone, and SUB along the growing menu provides the same rejection-permanence that IRC provides across menu sizes. The block is thereby ruled out without IRC.

\section{Conclusion}
The two relaxations of IRC differ in character. In the many-to-many market, IRP is
logically independent of IRC---an incumbency condition that neither implies nor is implied
by menu-invariance---and it drives existence by making acceptance permanent while SUB makes
rejection permanent. In the one-to-one market, BA is a strictly weaker condition than IRC
under SUB: it constrains only strict pairwise choice, forbidding cycles without imposing
full menu-consistency, and SUB turns the resulting acyclic relation into an order along
which deferred acceptance yields an R-stable matching. The many-to-one market combines the
two, and existence holds throughout outside path independence.

Several questions remain open. It is unknown whether the set of CY-stable matchings under
IRP retains the lattice structure associated with PI, and how IRP and BA extend to markets
with capacities or with agents arriving over time.\footnote{\cite{alkan2003stable} study
multi-valued choice in matching and extend \cite{alkan2001preferences} and
\cite{alkan2002class} to show that stable matchings exist and that the set of stable
matchings forms a lattice.}

Incentive properties raise a further set of open questions, and the three settings differ
in what to expect. The symmetric many-to-many algorithm has no proposing side, so the
proposer-optimality behind classical strategy-proofness has no analogue, and we expect it
to be manipulable.\footnote{For a capacity-one agent, whose assignment is a single
contract, the true correspondence induces via BA a strict order $\succ_i$; a report is an
admissible (BA) correspondence $\hat C_i$, and $i$ manipulates if the contract the induced
order yields is $\succ_i$-preferred to the truthful one. For a capacity-many agent $f$,
whose assignment is a set, one may instead compare bundles by revealed choice---$f$ prefers
$T'$ to $T$ if $C_f(T\cup T')=T'$---a comparison that is only partial and coincides with a
ranking over sets exactly when $C_f$ satisfies IRC \citep{yang2020rationalizable}.} The
one-to-one and many-to-one procedures are deferred-acceptance algorithms, for which one
anticipates strategy-proofness for the proposing side and, by the impossibility of
two-sided strategy-proofness \citep{roth1982economics}, manipulability for the other. When
firms' choices are responsive, the many-to-one procedure reduces to classical
worker-proposing deferred acceptance, so workers cannot profitably misreport; whether this
extends to all substitutable, IRP firm correspondences, and whether IRP's permanent
acceptance strengthens or weakens the proposing side's guarantees, is open.

\section*{Acknowledgements} The authors would like to thank Jacob Coreno, Bhaskar Dutta, Debasis Mishra, Arunava Sen, and Rohit Vaish for comments and suggestions. The authors acknowledge the use of Gemini and Claude AI to improve the readability of the manuscript and to help refine the presentation of the mathematical proofs and algorithms. The authors take full and sole responsibility for the original ideas, theoretical results, and the final content of this manuscript. No data was used in the paper. The authors declare that there are no conflicts of interest and no external funding was received for this research.

\bibliography{biblio.bib}

\appendix

\section*{Appendix}

\section{The many-to-many matching problem}

\subsection{Proof of Theorem \ref{thm:stability}}

\begin{proof}
We establish the theorem through a sequence of lemmas.

\begin{lemma}\label{lem:invariant}
Suppose $C_i$ satisfies SUB and IRP. Then $H_i(t-1) \subseteq H_i(t)$ for every $t$ and every agent $i$.
\end{lemma}

\begin{proof}
$H_i(t)$ is IR for every $t$: $H_i(0) = \emptyset$ trivially, and whenever $H_i$ updates, it is set to $C_i(\cdot)$ applied to some set, which is IR by SUB (Observation~\ref{obs:subcc}(ii)).

If $i$ is not an endpoint of $x_t$, $H_i(t) = H_i(t-1)$, trivially. If $i$ is an endpoint and $x_t$ is rejected, $H_i(t) = H_i(t-1)$, trivially. If $i$ is an endpoint and $x_t$ is accepted, then in particular $x_t \in C_i(H_i(t-1) \cup \{x_t\})$, i.e.\ $x_t$ is a challenger to the IR set $H_i(t-1)$. IRP applies directly: $H_i(t-1) \subseteq C_i(H_i(t-1) \cup \{x_t\}) = H_i(t)$.
\end{proof}

Permanence of rejection is derived from this using SUB as follows. 

\begin{observation}\label{obs:permrej}
Suppose $C_i$ satisfies SUB and IRP. If $x_t = (f,w)$ is rejected at round $t$ because $x_t \notin C_i(H_i(t-1) \cup \{x_t\})$ for some $i \in \{w,f\}$, then $x_t \notin C_i(H_i(t') \cup \{x_t\})$ for every $t' \ge t$; in particular $x_t \notin H_i(t')$ for every $t' \ge t$.
\end{observation}

\begin{proof}
By hypothesis, $x_t \notin C_i(H_i(t-1) \cup \{x_t\})$. By Lemma~\ref{lem:invariant}, $H_i(t-1) \subseteq H_i(t')$ for $t' \ge t$, so $H_i(t-1) \cup \{x_t\} \subseteq H_i(t') \cup \{x_t\}$, and SUB gives $x_t \notin C_i(H_i(t') \cup \{x_t\})$.
\end{proof}

Thus every rejection is permanent, regardless of whether it originates from the firm or the worker. Once a contract is rejected, it is never reconsidered.

\begin{lemma}\label{lem:srirc} 

If $x \notin \nu^\ast$, then $x \notin C_w(\nu_w^\ast \cup \{x\})$ or $x \notin C_f(\nu_f^\ast \cup \{x\})$ (or both), where $w = w(x)$, $f = f(x)$.

\end{lemma}

\begin{proof}
Since $x \notin \nu^\ast$, $x$ was rejected at the single round $t$ it was processed (every contract is processed exactly once): $x \notin C_i(H_i(t-1) \cup \{x\})$ for some $i \in \{w,f\}$. Observation \ref{obs:permrej} gives $x \notin C_i(H_i(n) \cup \{x\}) = C_i(\nu_i^\ast \cup \{x\})$ for that same $i$.
\end{proof}

\emph{Individual rationality.} Whenever $H_i(t)$ updates, it is set to $C_i(\cdot)$ applied to some set IR by Observation~\ref{obs:subcc}(ii); it is unchanged otherwise. Hence $\nu_i^\ast = H_i(n)$ is IR for every agent.

\emph{No blocking set.} Suppose $Z \ne \emptyset$ blocks $\nu^\ast$. Pick any $x \in Z$, and let $w = w(x)$, $f = f(x)$. Since $x \in X_w \cap X_f$, $x \in Z_w$ and $x \in Z_f$. Blocking gives $Z_w \subseteq C_w(\nu_w^\ast \cup Z_w)$ and $Z_f \subseteq C_f(\nu_f^\ast \cup Z_f)$, so Lemma~\ref{lem:reduction} applied to each side gives
\[
x \in C_w(\nu_w^\ast \cup \{x\}) \qquad \text{and} \qquad x \in
C_f(\nu_f^\ast \cup \{x\}).
\]
Since $Z \cap \nu^\ast = \emptyset$, $x \notin \nu^\ast$, so
Lemma~\ref{lem:srirc} gives
\[
x \notin C_w(\nu_w^\ast \cup \{x\}) \qquad \text{or} \qquad x \notin
C_f(\nu_f^\ast \cup \{x\}),
\]
contradicting one of the two statements above. No blocking set exists.
\end{proof}

\subsection{IRC and IRP are logically independent}\label{sec:independence}

Since PI is equivalent to SUB together with IRC, and both examples satisfy SUB, the substantive comparison is between IRC and IRP. We show these are logically independent: neither implies the other.

\begin{proposition}[Independence of IRC and IRP]\label{prop:irc-irp}
There exist choice correspondences satisfying SUB and IRC but not IRP, and correspondences satisfying SUB and IRP but not IRC.
\end{proposition}

\begin{proof}
\emph{SUB and IRC without IRP.} Let $X = \{a,b\}$ with $C(\{a\}) = \{a\}$, $C(\{b\}) = \{b\}$, $C(\{a,b\}) = \{b\}$. SUB holds: the two non-trivial pairs are $C(\{a,b\}) \cap \{a\} = \emptyset \subseteq C(\{a\})$, vacuously, and $C(\{a,b\}) \cap \{b\} = \{b\} \subseteq C(\{b\}) = \{b\}$, by equality. IRC holds: the only pair $(A,B)$ with $A \subsetneq B$ and $C(B) \subseteq A$ is $A = \{b\} \subseteq B = \{a,b\}$ (since $C(B) = \{b\}$), and $C(\{b\}) = \{b\} = C(B)$, as required; every other pair either fails $C(B) \subseteq A$, so imposes no constraint, or has $A = B$, trivially. But IRP fails: $\{a\}$ is IR, trivially, and $b$ is a challenger to it ($b \in C(\{a,b\})$); IRP would require $\{a\} \subseteq C(\{a,b\}) = \{b\}$, which is false.

\emph{SUB and IRP without IRC.} Let $X = \{a,b\}$ with $C(\{a\}) = \{a\}$, $C(\{b\}) = \emptyset$, $C(\{a,b\}) = \emptyset$. SUB holds, checked directly on both non-trivial pairs. IRP holds vacuously: $\{a\}$ is the only non-empty IR  set in this domain, and $b \notin C(\{a,b\}) = \emptyset$,
so $b$ is never a challenger to it; IRP's hypothesis is never triggered. But IRC fails: $C(\{a,b\}) = \emptyset \subseteq \{a\} \subseteq \{a,b\}$ requires $C(\{a\}) = \emptyset$, which is false since $C(\{a\}) = \{a\}$.
\end{proof}

\section{The one-to-one matching problem}
\subsection{Proof of Theorem \ref{thm2}}

\begin{proof}
We introduce some notation and a lemma. For IR $x,y\in X_i$, write $x\mathrel{P_i}y$
(``$i$ reveals $x$ over $y$'') if $C_i(\{x,y\})=\{x\}$. By Axiom~\ref{ax:ba}, $P_i$ is
acyclic for every agent $i$.

\begin{lemma}\label{lem:acc}
If $x\in C_i(S)$ for some $S\subseteq X_i$, then $x$ is IR. In particular, if
$C_i(\nu_i\cup\{x\})=\{x\}$ then $x$ is IR for $i$.
\end{lemma}
\begin{proof}
Apply Observation~\ref{obs:subcc}(i) in contrapositive form with $\{x\}\subseteq S$: since
$x\in C_i(S)$ and $x\in\{x\}$, we get $x\in C_i(\{x\})$.
\end{proof}

\textbf{Preferences.}
Fix an agent $i$. By Axiom~\ref{ax:ba}, $P_i$ is acyclic, so its transitive closure is a
strict partial order on $i$'s IR contracts; extend it to a strict linear order $\succ_i$ on
those contracts that preserves every strict comparison
($x\mathrel{P_i}y\Rightarrow x\succ_i y$). On pairs with
$C_i(\{x,y\})\in\{\varnothing,\{x,y\}\}$, $P_i$ ranks neither contract over the other, and
$\succ_i$ resolves them arbitrarily, with $\varnothing$ ranked below every IR contract as
the outside option. The resolution is immaterial: the no-blocking argument uses $P_i$ only
through the implication $C_i(\{x,\nu_i\})=\{x\}\Rightarrow x\succ_i\nu_i$, which depends only
on strict comparisons.

\textbf{Matching.}
Run the firm-proposing deferred acceptance algorithm on the strict linear orders
$(\succ_f)_{f\in F},(\succ_w)_{w\in W}$, breaking ties as fixed above
\cite{roth2008deferred}. Since each agent holds at most one contract, this yields a
one-to-one matching $\nu$. By the stability of the deferred-acceptance outcome
\cite{gale1962college}, $\nu$ has no \emph{$\succ$-blocking pair}: no $x=(f,w)\notin\nu$, IR
for both endpoints, satisfies $x\succ_f\nu_f$ and $x\succ_w\nu_w$ (with $\nu_i=\varnothing$
read as ``$i$ unmatched,'' which any IR contract beats).

\textbf{Individual rationality.}
Firms propose, and workers accept, only IR contracts, so each $\nu_i$ is either $\varnothing$
or a singleton $\{x\}$ with $x$ IR. If $\nu_i=\{x\}$ then $C_i(\{x\})=\{x\}$; if
$\nu_i=\varnothing$ then $C_i(\varnothing)=\varnothing$. Hence $C_i(\nu_i)=\nu_i$ for every
$i$, and $\nu$ is IR.

\textbf{No blocking.}
Suppose, for contradiction, that $x=(f,w)\notin\nu$ blocks $\nu$, i.e.\
$C_f(\nu_f\cup\{x\})=\{x\}$ and $C_w(\nu_w\cup\{x\})=\{x\}$. By Lemma~\ref{lem:acc}, $x$ is
IR to both $f$ and $w$. On the firm side: if $\nu_f=\varnothing$, then $x$ is IR and $f$ is
unmatched, so $x\succ_f\nu_f$; if $\nu_f=\{y\}$ with $y\ne x$, then $C_f(\{y\})=\{y\}$ (IR)
and $C_f(\{x,y\})=\{x\}$ gives $x\mathrel{P_f}y$, so $x\succ_f y=\nu_f$ since $\succ_f$
extends $P_f$. Either way $x\succ_f\nu_f$; the worker side is symmetric, giving
$x\succ_w\nu_w$. Thus $x$ is a $\succ$-blocking pair for $\nu$, contradicting the absence of
such a pair. Hence no blocking contract exists, and $\nu$ is R-stable.
\end{proof}

\subsection{BA is weaker than IRC under SUB}\label{secb2}

We show that BA is weaker than IRC under SUB.

\begin{proposition}\label{prop:ba-irc}
Under SUB, IRC implies BA. The converse fails: there is a correspondence satisfying SUB
and BA but not IRC. Hence, on the substitutable domain, BA is strictly weaker than IRC.
\end{proposition}
\begin{proof}
\emph{SUB and IRC imply BA.} Under SUB, IRC is equivalent to path independence
\cite{aizerman1981general}, and every path-independent choice correspondence is
rationalizable by a linear order over sets \cite{yang2020rationalizable}. Such an order
induces a complete, transitive ranking of singletons, whose induced strict pairwise
relation $P$ is acyclic; hence BA holds.

\emph{BA does not imply IRC.} On $X=\{a,b,c\}$ let
\[
C(\{a\})=\{a\},\ C(\{b\})=\{b\},\ C(\{c\})=\{c\},\quad
C(\{a,b\})=\{a\},\]\[ C(\{a,c\})=\{a,c\},\ C(\{b,c\})=\{c\},\ C(\{a,b,c\})=\{a\}.
\]
SUB holds (direct check). The only strict pairwise preferences are $a$ over $b$ and $c$
over $b$, which form no cycle, so BA holds. IRC fails:
$C(\{a,b,c\})=\{a\}\subseteq\{a,c\}\subsetneq\{a,b,c\}$, yet $C(\{a,c\})=\{a,c\}\neq\{a\}$.
\end{proof}

\section{The many-to-one matching problem}

\subsection{Proof of Theorem \ref{thm:mo}}
\begin{proof}

The following two structural facts about firms carry over from the many-to-many analysis, since each firm receives contracts one at a time.

(i) \emph{Monotone firm holdings:} Whenever $H_f$ updates it is set to $C_f(\cdot)$, hence IR by Observation~\ref{obs:subcc}(ii); the added contract $x$ satisfies $x\in C_f(H_f\cup\{x\})$, so $x$ is a challenger to the IR set $H_f$, and IRP gives $H_f\subseteq C_f(H_f\cup\{x\})$. Holdings therefore only grow (Lemma~\ref{lem:invariant}, firm side).

(ii) \emph{Permanent rejection:} If $x$ is rejected by $f$ at holding $H$ ($x\notin C_f(H\cup\{x\})$), then for any later $H'\supseteq H$ we have $H\cup\{x\}\subseteq H'\cup\{x\}$, and Observation~\ref{obs:subcc}(i) gives $x\notin C_f(H'\cup\{x\})$; in particular $x\notin C_f(\nu_f\cup\{x\})$ (Observation~\ref{obs:permrej}, firm side).

\emph{Individual rationality} - Each $\nu_f=H_f$ is a chosen set, so $C_f(\nu_f)=\nu_f$ by Observation~\ref{obs:subcc}(ii). Each matched worker holds a single contract it proposed; since $w$ proposes only along $\succ_w$, whose contracts are IR by construction, that contract is IR, so $C_w(\nu_w)=\nu_w$. Unmatched workers have $\nu_w=\emptyset$. Thus $\nu$ is IR.

\emph{No blocking} - Suppose $x=(f,w)\notin\nu$ MO-blocks $\nu$; by Definition~\ref{def:mo} this means $x\in C_f(\nu_f\cup\{x\})$ (firm-side inclusion) and $C_w(\nu_w\cup\{x\})=\{x\}$ (worker-side replacement). The worker-side equality gives $x\succ_w\nu_w$: if $\nu_w=\{y\}$ then $C_w(\{x,y\})=\{x\}$, so $x\,P_w\,y$ and $\succ_w$ extends $P_w$; if $\nu_w=\varnothing$ then $x$ is IR (Lemma~\ref{lem:acc}) and precedes the outside option. Since $w$ proposes in $\succ_w$-order and $x\succ_w\nu_w$, $w$ proposed $x$ to $f$ before settling for $\nu_w$. Firm acceptance is permanent, so any contract $f$ ever accepts remains in its holding; as $x\notin\nu$, $f$ never accepted $x$, and hence rejected it at the round it was proposed---that is, $x\notin C_f(H\cup\{x\})$ for $f$'s holding $H$ at that round. By permanent rejection (established above), $x\notin C_f(H'\cup\{x\})$ for every later holding $H'\supseteq H$, and in particular $x\notin C_f(\nu_f\cup\{x\})$. This contradicts the firm-side block condition $x\in C_f(\nu_f\cup\{x\})$. Hence no such $x$ exists, and $\nu$ is MO-stable.
\end{proof}

\section{A comparison of IRC and IRP (under SUB)}\label{app:comparison}

We already know that every path-independent choice correspondence is rationalizable by a preference over sets of contracts (\cite{yang2020rationalizable}). However, this is not true for choice correspondences which satisfy SUB and IRP. We provide an example.

\begin{example}
\label{ex:no-set-ranking}
Let $X=\{a,b,c\}$.

\emph{A path-independent correspondence.} Let $C$ be responsive with quota two under
$a\succ b\succ c$:
\[
\begin{aligned}
&C(\{a\})=\{a\},\quad C(\{b\})=\{b\},\quad C(\{c\})=\{c\},\\
&C(\{a,b\})=\{a,b\},\quad C(\{a,c\})=\{a,c\},\quad C(\{b,c\})=\{b,c\},\quad C(\{a,b,c\})=\{a,b\}.
\end{aligned}
\]

This correspondence satisfies SUB and IRC, hence PI, and is rationalized by the ranking over subsets

\[
\{a,b\}\;\succ\;\{a,c\}\;\succ\;\{b,c\}\;\succ\;\{a\}\;\succ\;\{b\}\;\succ\;\{c\}\;\succ\;\varnothing,
\]

in the sense that $C(S)$ is the $\succ$-greatest subset contained in $S$. Deleting a rejected contract never changes the chosen set---$C(\{a,b,c\})=C(\{a,b\})=\{a,b\}$, for instance---which is precisely what allows a single ranking of subsets to reproduce $C$.

\emph{A correspondence satisfying SUB and IRP but not IRC.} Let
\[
\begin{aligned}
&C(\{a\})=\{a\},\quad C(\{b\})=\{b\},\quad C(\{c\})=\{c\},\\
&C(\{a,b\})=\{a,b\},\quad C(\{a,c\})=\{a,c\},\quad C(\{b,c\})=\varnothing,\quad C(\{a,b,c\})=\{a\}.
\end{aligned}
\]
This correspondence satisfies SUB and IRP but violates IRC: $C(\{a,b,c\})=\{a\}\subseteq  \{a,b\}\subsetneq \{a,b,c\},$ yet $ C(\{a,b\})=\{a,b\} \neq \{a\}.$ No ranking over subsets rationalizes it since $C(\{b,c\})=\varnothing$ would require $\varnothing\succsim\{b\}$ as $\{b\}\subseteq\{b,c\}$. However, $C(\{b\})=\{b\}$ requires $\{b\}\succsim\varnothing$. No fixed preference over subsets can select $\varnothing$ from $\{b,c\}$ and $\{b\}$ from $\{b\}$.
\end{example}

\end{document}